\documentclass[acmsmall,nonacm,screen]{acmart}
\usepackage{subcaption}
\usepackage{physics}
\usepackage{xfrac}
\usepackage{thmtools, thm-restate}
\usepackage{tikz}
\usetikzlibrary{quantikz2}

\AtBeginDocument{%
  }

\title{A Denotational Semantics for Quantum Loops}

\author{Nicola Assolini}
\email{nicola.assolini@univr.it}
\affiliation{%
  \institution{University of Verona}
  \city{Verona}
  \country{Italy}
}
\author{Alessandra Di Pierro}
\email{alessandra.dipierro@univr.it}
\affiliation{%
  \institution{University of Verona}
   \city{Verona}
  \country{Italy}
}


\begin{document}

\renewcommand{\textasteriskcentered}{*}

% MACRO
%%%%%%%
\newcommand*{\ud}{\:=\:}
\newcommand*{\df}{\:\Coloneqq\:}
\newcommand*{\arrw}[1]{\overline{#1}}
\newcommand*{\roundB}[1]{(#1)}
\newcommand*{\namRule}[1]{(\textbf{#1})}
\newcommand*{\quotes}[1]{`#1'}

%% set and operators
\newcommand*{\Hilb}{\mathcal{H}}
\newcommand*{\Dom}{\Hilb_P}
\newcommand*{\Fock}{\mathcal{F}}
\newcommand*{\Gard}{\mathcal{G}}
\newcommand*{\Uni}{\mathcal{U}}
\newcommand*{\Bnd}{\mathcal{B}}
\newcommand*{\UniSub}{\mathcal{U}\mathcal{S}}
\newcommand*{\Opr}{\mathcal{O}}
\newcommand*{\Ident}{\mathbf{I}}
\newcommand*{\bev}{\textbf{G}}
\newcommand*{\sumshort}{\mathbb{S}}
\newcommand{\Cfun}{\mathcal{C}}
\newcommand{\lub}{l.u.b.}

\newcommand*{\cS}{\mathcal{S}}
\newcommand*{\cA}{\mathcal{A}}
\newcommand*{\cL}{\mathcal{L}}

\newcommand*{\CN}{CNOT}

% \newcommand*{\Dset}{\mathcal{I}}

%% program's keyword
\newcommand*{\qift}[2]{\texttt{if}\,#1\,\texttt{do}\ \{#2\}}
\newcommand*{\whle}[2]{\texttt{while}\ #1\,\texttt{do}\ \{#2\}}
\newcommand*{\whleB}[3][k]{\texttt{while}^{#1}\ #2\,\texttt{do}\ \{#3\}}
\newcommand*{\wt}{\texttt{w}}
\newcommand*{\Wt}{\texttt{W}}
\newcommand*{\lt}{\texttt{l}}
\newcommand*{\Lt}{\texttt{L}}
\newcommand*{\skp}{\texttt{skip}}

% prog semantics
\newcommand*{\lnSem}[1]{{\llbracket\,#1\,\rrbracket}}
\newcommand*{\unSem}[1]{{[\,#1\,]}}
\newcommand*{\prj}[1]{\texttt{P}_{#1}}
\newcommand*{\inferr}[3]{\inference[(\textbf{#1})\:]{#2}{#3}}
\newcommand*{\sRule}[3][]{#2\!\rightarrow^{#1}\!#3}
\newcommand*{\opSem}[2]{[#1\ ,\, #2]}

%% expression words
% \newcommand*{\choice}[2]{#1 \oplus #2}
% \newcommand*{\iter}[2][*]{#2^{#1}}

% Ming command
% \newcommand*{\qcase}[3][P]{\texttt{qif}\,[#2](\square #3 \cdot \qctrl{#3}{#1})}
% \newcommand*{\qctrl}[2]{\ket{#1} \rightarrow {#2}_{#1}}

\newcommand{\TODO}[1][]{{\color{red}[TODO: #1]}}
\newcommand*{\ciao}[3][7]{(#1,#2,#3)}

\begin{abstract}
    Programming a quantum computer, i.e., implementing quantum algorithms on a quantum processor-based copmputer architecture, is a task that can be addressed (just as for classical computers) at different levels of abstraction.
    This paper proposes a denotational semantics for high-level quantum programming constructs, focusing on the conceptual meaning of quantum-controlled branching and iteration. 
    We introduce a denotational domain where a mathematical meaning of a quantum control flow with loops can be defined, which reflects the coherent evolution of the quantum system implementing the program.
\end{abstract}

\maketitle

\section{Introduction}
\label{intro}
A crucial part of a computer program is its control flow.
In classical computing, control flow refers to the sequencing and branching of instructions within a program, enabling the program to make decisions and alter its behavior based on specific conditions.
This also implies the possibility of writing programs with conditional loops.

The control flow of quantum programs cannot be interpreted in the same way as for classical programs due to the properties of the target physical device (where they are intended to be executed), which behaves according to the laws of quantum mechanics.
Notably, a quantum processor is able to work on \textit{superpositions} of states (qubits) rather than on single ones and, in a more strikingly different way from a classical computer, it can generate states which are \textit{entangled}, i.e., tied to each other by a strong (non-classical) correlation.
Moreover, while the results of the execution of a classical program are immediately available whenever the program reaches the final statement, accessing the results of a quantum program is not so straightforward, due to the so-called measurement problem in the theory of quantum mechanics. In fact, although quantum theory is, up to now, the most precise description of how the world behaves, the interpretation of such a behaviour is controversial and the debate on which, among the several interpretations that have been proposed, is the right one is still open and represents the main problem for a complete understanding of quantum physics.

When analysing a quantum program and studying its mathematical behaviour, it is inevitable to refer to the interpretation of quantum mechanics. Typically, the quantum programming language literature refers to one of the most accredited interpretations, which goes under the name of the Copenhagen interpretation. According to this interpretation, the results of the execution of a quantum program can only be obtained when the coherent (i.e. in superposition) execution (or evolution of the quantum system) collapses into a classical state, which occurs randomly both in time (at a given average rate), and in space (according to the Born rule). This explanation avoids the measurement problem and leads to modelling the result of a quantum program essentially as a probability distribution on all its possible outcomes.

In this paper, we aim at a description of a quantum program that is as general as possible by avoiding an explicit syntactic construct for measurement, which would effectively lead to a (classical) probabilistic semantics of the program.
%\subsubsection*{Our contribution}
Moreover, in giving a meaning to a quantum program, we will concentrate on the denotational approach, following Strachey's observation that fixing the domain within which programs in a given language have their meanings tell us a great deal about the language, and is a sure guide to the design of the language \cite{Stoy}. 
%We define a denotational semantics for a minimal high-level quantum programming language, which captures the unitary essence of quantum computation and avoids any explicit reference to the probabilistic interpretation of Quantum Mechanics given by the Copenhagen interpretation.
%To this aim, we will not refer in our treatment to the common assumption of the probabilistic interpretation of Quantum Mechanics (so-called Copenhagen interpretation), according to which a quantum state is doomed to a `collapse' generating probabilistic results according to the Born rule.

%\section{Quantum Loop}\label{sec: quantum loop}
At this point, a crucial question arises: \textit{how can we handle termination in quantum loops?}
In the Turing machine model of classical computation, we can use a \textit{halting bit} to signal termination.  
However, as analyzed in~\cite{linden1998halting,miyadera2005halting,myers1997universal,shi2002remarks,song2008unsolvability}, defining such a halting qubit is not possible is not possibile in a Turing machine model for quantum computation, without compromising the all computation.  
Unlike classical loops, where execution can be stopped based on the guard's value, measuring a quantum guard would make a quantum superposition collapse to a classical state, thus altering the computation itself.  
Quantum languages that rely on measurement-based control flow circumvent this issue but at the cost of introducing non-unitary behavior.  

An even deeper problem, also highlighted in the works mentioned above, is that a quantum loop can terminate on some inputs while diverging on others, leading to a superposition of terminating and non-terminating states.  
% Allowing such a superposition contradicts fundamental quantum mechanical principles.  
Since quantum computation is performed by unitary operators transforming quantum states (or superpositions) into quantum states,
to determine from the outside whether a quantum execution has reached a fixed point, one should either arbitrarily interrupt the computation after a finite number of steps, thus making it terminate on all inputs,  or let it evolve indefinitely.
If restricting our consideration to finite computations may seem a solution, it is actually not a very satisfactory one, as it would prevent a suitable definition of a semantics for formally reasoning about termination.   
%This has fundamental implications for the semantics of quantum loops.
, %we cannot define a unitary operator that represents a quantum loop, meaning that a purely quantum implementation of loops is inherently impossible.  
%which can only describe computations that either terminate on all inputs or correspond to an infinite sequence of transformations.  
%However, even if we considered sequences of unitary operators representing only the finite executions of a loop, it would be not straightforward to mathematically model such a behaviour as there is no intrinsic limit to this sequence.

%Besides, even if we considered sequences of unitary operators representing only the finite executions of a loop, it would be not straightforward to mathematically model such a behaviour as there is no intrinsic limit to such sequences.

%Moreover, while in practice it is not possible to have a superposition between termination and non-termination, a model that characterizes which parts of the execution terminate and which do not is still meaningful. 
%
Therefore, just like for the calssical case, having a semantics that captures also infinite computation is crucial for correctly modeling quantum programs, although in the quantum case, achieving this result is more problematic for the reasons discussed above. 
In this paper we show how a model that characterizes which parts of the execution terminate and which do not in a quantum loop can be defined by approximating the concrete unitary behaviour of a quantum program by some linear operators which are able to `separate' executions that have reached a fixed point from those that are still computing even in the limit, i.e. in the case of an infinite loop. In fact, by relaxing the unitarity constraint, we are able to construct an approximating sequence of linear operators which converges to a mathematically well-defined limit. While for finite computations, this corresponds exactly to the actual (unitary) behaviour of a quantum program, in the presence of infinite computations, it is a linear operator with a norm strictly less than one due to the portion of the superposition on which computation is still going on.  
%
%In this way, we can effectively characterize both termination and infinite computation, even in the presence of superposition.

\section{Background}
Programs in a quantum programming language are designed to run on quantum computers and are very different from classical computer programs. 
The design and implementation of such languages requires a sound knowledge of the principles of quantum mechanics and the underlying mathematics.
In this section, we briefly recall the main aspects of quantum computation that make this computational model different from the classical one. We will refer to the circuit model of computation and highlight such differences in terms of the meaning of wires and gates in a classical and a quantum circuit. 

In a quantum circuit, wires represent quantum bits, or qubits, rather than bits. 
The classical unit of information (the bit) generates, with its two values $0$ and $1$, a complex vector space (a quantum system), where each complex vector of norm $1$ is the state of a qubit. 
This is, therefore, a linear combination of the form $\ket{\psi} = \alpha \ket{0} + \beta \ket{1}$, where $\alpha$ and $\beta$ are complex numbers from which we can infer the probability of the state resulting (after measurement) in $1$ or $0$, respectively. 
Such probabilities are obtained as $|\alpha |^2$ and $| \beta |^2$, which explains why quantum states must be normalized vectors, i.e. $|\alpha |^2 + | \beta |^2 = 1$ must hold\footnote{The \textit{ket} notation $\ket{\psi}$ is due to Dirac and represents the vector $(\alpha,\beta)^T$ in linear algebraic notation.}. Such vectors live in a complex Hilbert space, equipped with the $\ell^2$ norm~\footnote{The $\ell^2$ norm of a vector $\ket{\psi} = [\alpha_1, \alpha_2, \dots, \alpha_n]^T$ is defined as $ \|\ket{\psi}\| = \Sigma_{i=1}^{n} \sqrt{|\alpha_i|^2}$.}~\cite{heinosaari2008guidemathematical}.  
The state of $n$ qubits corresponds to a unit vector in the  $2^n$-dimensional Hilbert space ($\Hilb^{2^n}$) obtained by composing by tensor product the normalized states corresponding to each wire (qubit), i.e. a vector in a 2-dimensional complex Hilbert space ($\Hilb^2$) \cite[Chapter 2]{MichealANielsen}.
For instance, the space of two qubits is $\Hilb^4 = \Hilb^2 \otimes \Hilb^2$ and a generic state $\ket{\psi}$ in $\Hilb^4$ can be written as $
\ket{\psi} = \alpha_0\ket{00} + \alpha_1\ket{01} + \alpha_2\ket{10} + \alpha_3\ket{11}$, where all $\alpha_i$ are complex numbers.

Throughout the paper, we will adopt the following convention.
For a variable $q$, we will write $\ket{\psi}_q$ to indicate that $q$ represents the state $\ket{\psi}$ of a qubit register. 
For multi qubits states, such as for example $\sfrac{1}{\sqrt{2}}(\ket{01} + \ket{10}$, we will write $\sfrac{1}{\sqrt{2}}(\ket{01} + \ket{10})_{p,q}$ to indicate that variable $p$ represents the first qubit and $q$ represents the second qubit of the entangled pair.

% \begin{definition}\label{def: opNorm}
%     Let $\|\psi\|$ denote the vector norm in $\Hilb$ and $A$ be an operator from $\Hilb$ to $\Hilb$, we define $\| A \| = \sup_{\|\psi\| = 1} \| A \psi \|$ as the operator norm.
% \end{definition}

\section{Quantum while language}\label{sec:psyntax}
To present our semantics, we introduce a generic quantum language with a minimal set of constructs consisting of  quantum loop iteration, sequential composition, and unitary transformation.
The syntax of our simple language, which we will refer to as $SL$, is given by the following grammar defining a program $s$ as follows:
\begin{equation}\label{eq: syntax}
    \begin{aligned}
        s \df&\ U(\arrw{q})\ \big|\ s; s\ \big|\ \skp\ \big|\ \whle{q}{s},
    \end{aligned}
\end{equation}
where $q$ is a quantum variable and $\arrw{q}$ denotes a sequence $q^1,\dots,q^{n}$ of quantum variables.
In $SL$, we do not define a command for state's initialization (or assignment) since we assume that all variables are initialized to $\ket{0\dots0}$.
All operations on variables are performed by statement $U(\arrw{q})$ corresponding to applying a unitary transformation $U$ to $\arrw{q}$. 
Every possible quantum transformation, except measurement, is a composition of unitary transformations.
The principle of deferred measurement~\cite[Chapter 4]{MichealANielsen} guarantees that all intermediate measurements in a quantum circuit can be moved to the end of the computation; thus, excluding measurement operation from our languages does not cause a loss of generality. Moreover, as explaind in Section~\ref{intro}, this will make our treatment independent of any specific interpretation of quantum mechanics.
Since measurement does not occur in $SL$ programs, their semantics is a deterministic transition between quantum states according to the physical law (the Shrödinger equation) governing the quantum system on which the program is intended to be executed, without having to establish when quantum mechanics should leave space to classical mechanics (which essentially constitutes the measurement problem in quantum mechanics).
%Probabilistic outcomes are only obtained by measuring the resulting state at the end of the program.

\subsection{The State Space}
To define the semantics of $SL$, we consider a domain similar to the one described in~\cite[Chapter 3]{MingFoundations}.
Given a program $s$, let $ Q_{s}$ be the set of all variables occurring in $s$.
Each quantum variable $q \in Q_s$ has a type $\Hilb_q$, and it is interpreted as a vector $\ket{\varphi}_q$ in its own Hilbert space $\Hilb_q$.
For example, for a Boolean variable $b$, the state $\ket{\varphi}_b = \alpha\ket{0} + \beta\ket{1}$, $\alpha,\beta \in \mathbb{C}$ corresponds to a vector in the complex Hilbert space of dimension $2$ (the state space of 1 qubit), which is the type of $b$.
We define
\begin{equation}
    \Hilb_{Q_s} = \bigotimes_{q\in Q_s} \Hilb_q,
\end{equation}
as the space of the types of all variables in $Q$.
We consider an infinitely countable set $T = \{t_i\}$ of ancillary quantum boolean variables, which are necessary to perform quantum while loops, and their Hilbert space:
\begin{equation}
    \Hilb_{T} = \bigotimes_{t_i \in T} \Hilb_{t_i} \text{ where } \Hilb_{t_i} = \Hilb^2
\end{equation}
Finally, we call $\Dom = \Hilb_{T} \otimes \Hilb_{Q_s}$ the Hilbert space of the program $s$.

\section{Unitary Semantics}
A quantum language with no classical operations can be completely described by using unitary operators that can be visually represented by quantum circuits.
A mathematical description of these circuits is by means of linear algebra and, in particular, by the group of unitary operators on a Hilbert space.

We start by defining the semantics for each statement in $SL$ except the $\texttt{while}$ statement. 
\begin{definition}\label{def:unSem}
    Let $\Uni(\Dom)$ be the group of unitary operators from $\Dom$ to $\Dom$.
    The unitary semantics is the function $\unSem{\cdot}: s \to \Uni(\Dom)$, defined by:
\begin{enumerate}
    \item $\unSem{\skp} = \Ident_{\Dom}$;
    \item $\unSem{U(\arrw{q})} = \arrw{U}$, where $\arrw{U} = \Ident \otimes U \otimes \Ident$, i.e. the extension of $U$ on $\Dom$;
    \item $\unSem{s_1;\, s_2} = \unSem{s_2}\cdot\unSem{s_1}$.
\end{enumerate}
\end{definition}
Rule \namRule{2} corresponds to computing the operator $U$ on $\arrw{q}$, and the identity on the other components of the program space. Rule \namRule{3} models sequential composition by means of matrices multiplication (the group operation).
Since $\arrw{U}$ and $\Ident$ are unitary operators, the semantics of each statement is a unitary operator.
Therefore, every statement can be represented by a circuit as shown in \autoref{fig: Uq and seq}.
% , for the rules \namRule{1} and \namRule{3}.
The case for the while is slightly more involved, and we treat it in the next section.

\begin{figure}
    \centering
    \begin{subfigure}[b]{0.48\linewidth}
        \centering
        \begin{quantikz}[wire types={q,n,q,b},classical gap=0.07cm,row sep=0.3em, column sep = 0.7em]
            \lstick[2]{$T$} \ t_{1} \ & & &\\
            % \ t_{1} \ & & &\\
            \vdots & & & \\
            \lstick[2]{$Q_s$} \ \overline{q} \ & \qwbundle{n}  & \gate{U} & \\
            & & & \\
        \end{quantikz}
        \centering
        \caption{}
        \label{fig: Uq}
    \end{subfigure}
    \begin{subfigure}[b]{0.48\linewidth}
        \centering
        \begin{quantikz}[wire types={q,n,b},classical gap=0.07cm,row sep=0.3em, column sep = 0.7em]
            \lstick[2]{$T$} \ t_{1} \ & \gate[3]{s_1}  & \gate[3]{s_2} &\\
            \vdots & & & \\
            \lstick{$Q_s$} & & & \\
        \end{quantikz}
        \caption{}
        \label{fig: seq}
    \end{subfigure}
    \caption{The circuits corresponding to $U(\arrw{q})$~\textbf{(a)} and $s_1;s_2$~\textbf{(b)}.}
    \label{fig: Uq and seq}
\end{figure}
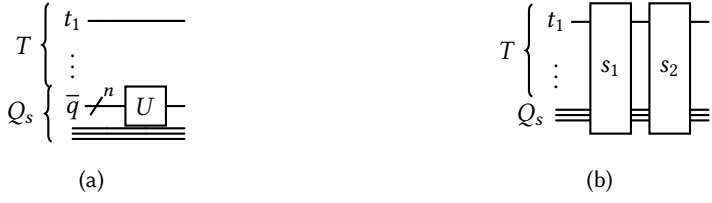

\subsection{While Loop Semantics}\label{sec: while sem}
A well-formed $\texttt{while}$ statement should modify its guard within the loop body.  
To define a quantum $\texttt{while}$ instruction, we need a representation where the guard qubit can be updated while preserving the unitarity of the evolution.  
Ideally, given a program $p = \qift{q}{U(q)}$ \footnote{We use here the `if' notation as a shortcut for indicating one iteration of the while statement}, we would like to have a unitary operator $C$ such that $C(\alpha\ket{0}_q + \beta\ket{1}_q) = \alpha\ket{0}_q + \beta U\ket{1}_q$.
However, in quantum computing, controlled operations cannot have their controller affected by the target, as this would break unitarity.  
For example, consider $p = \qift{q}{X(q)}$. The corresponding operator $C$ would map $\ket{0}_q$ to $\ket{0}_q$ and $\ket{1}_q$ to $\ket{0}_q$, which is clearly a non-injective, and thus non-unitary.
We must introduce an auxiliary qubit to model a self-controlled operation within a unitary framework on which to copy information via a \CN\ gate.  
\autoref{fig:ifCirc} illustrates a hypothetical quantum if-statement where the guard is included in the body.  
The semantics can be represented by a unitary operator $CU$ such that $CU(\ket{0}_t \otimes (\alpha\ket{0}_q + \beta\ket{1}_q)) = \alpha\ket{0}_t\ket{0}_q + \beta \ket{1}_t U\ket{1}_q$.  
Extending this approach to a while loop introduces an additional challenge: each iteration requires a fresh temporary qubit, which implies the need for an infinite set of auxiliary qubits to model non-terminating loops.  
\autoref{fig:whileCirc} shows the circuit representation of a quantum while loop, which is recursively defined and corresponds to an infinite composition of unitary operations on an infinite-dimensional Hilbert space.

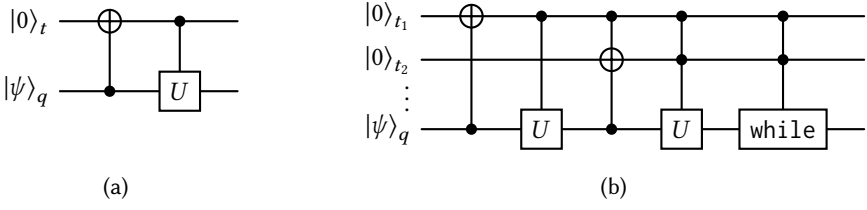
\begin{figure}
    \centering
    \begin{subfigure}[b]{0.3\linewidth}
        \centering
        \begin{quantikz}[wire types={q}]
            \lstick{$\ket{0}_t$} & \targ{} & \ctrl{1} & \\
            \lstick{$\ket{\psi}_q$} & \ctrl{-1} & \gate{U} & \\        
        \end{quantikz}
        \centering
        \caption{}
        \label{fig:ifCirc}
    \end{subfigure}
    \begin{subfigure}[b]{0.64\linewidth}
        \centering
        \begin{quantikz}[wire types={q,q,n,q},row sep=0.7em]
            \lstick{$\ket{0}_{t_{1}}$} & \targ{} & \ctrl{3} & \ctrl{1} & \ctrl{1} & \ctrl{1} & \\
            \lstick{$\ket{0}_{t_{2}}$} & & & \targ{} & \ctrl{2} & \ctrl{2} & \\
             \lstick{$\vdots$} & & & & & & \\
             \lstick{$\ket{\psi}_{q}$} & \ctrl{-3} & \gate{U} & \ctrl{-2} & \gate{U} & \gate{\texttt{while}} & 
        \end{quantikz}
        \caption{}
        \label{fig:whileCirc}
    \end{subfigure}
    \caption{Quantum circuits for $\qift{q}{U(q)}$ ~\textbf{(a)} and $\whle{q}{U(q)}$~\textbf{(b)}}
    % \label{fig:enter-label}
\end{figure}

The first step in defining the semantics of a while loop is to start with the controlled operation.
In general, consider a Hilbert space $ \Hilb = \Hilb_g \otimes \Hilb_s $, where $ \Hilb_g $ is a 2-dimensional Hilbert space representing the control (guard) qubit, and $\Hilb_s$ is the space of target qubits.
Given a unitary operator $U_s$ in $\Hilb_s$, the unitary operator corresponding to $U_s$ controlled by $g$ can be represented by $\ketbra{0}_g \otimes \Ident_s + \ketbra{1}_g \otimes U_s$, where $\Ident_s$ is the identity in $\Hilb_s$ and $\ketbra{0}_g$ and $\ketbra{1}_g$ are the projector on $\ket{0}_g$ and $\ket{1}_g$ respectively in $\Hilb_g$.
\begin{proposition}\label{prop: ctrl}
    Let $\prj{i_{g}} = (\ketbra{i}_g \otimes \Ident_s)$ be the projector $\ketbra{i}_g$ extended to the whole $\Hilb$,
    and let $\arrw{U} = ( \Ident_g \otimes U_s)$ be the extension of $U_s$ in $\Hilb$.
    The controlled operation $\ketbra{0}_g \otimes \Ident_s + \ketbra{1}_g \otimes U_s$ is equivalent to the operator $\prj{0_{g}} + \prj{1_{g}}\cdot\arrw{U}$.
\end{proposition}
\begin{proof}
    By the property $AC \otimes BD = (A \otimes B)(C \otimes D)$~\cite[Lemma 4.2.10]{roger1991topics} we have $\ketbra{1}_g \otimes U_s = (\ketbra{1}_g \cdot \Ident_g) \otimes (\Ident_s \cdot U_s) = (\ketbra{1}_g \otimes \Ident_s) ( \Ident_g \otimes U_s) = \prj{1_{g}}\cdot\arrw{U}$.
    Thus $\ketbra{0}_g \otimes \Ident_s + \ketbra{1}_g \otimes U_s = \prj{0_{g}} + \prj{1_{g}}\cdot\arrw{U}$.
\end{proof}

As shown in \autoref{fig:whileCirc}, to perform a quantum loop, we need to make a quantum \quotes{copy} of the guard variable in a new fresh ancillary variable for every iteration.
This is realized by a \CN\ gate, which we will denote by $\bev(q, n)$ in our semantics, where $q$ is a variable in $Q_s$ and $n \in \mathbb{N}$, $n \ge 1$ indicates the target $t_n \in T$.

Given a qubit $q \in Q_s$, and a unitary operator $S$ that does not act on $t_1$, we recursively define an operator $\wt_n(q, S)$, using the control operation in the format introduced in \autoref{prop: ctrl}, as follows:
\begin{equation}\label{eq: wnricdef}
    \begin{aligned}
        \wt_0(q,S) &= \Ident\\
        \wt_n(q,S) &= (\prj{0_{t_1}} + \prj{1_{t_1}} \cdot \cS(\wt_{n-1}(q,S)) \cdot S) \bev(q,1).
    \end{aligned}
\end{equation}
where $\cS$ shifts the controls to keep the first ancilla qubit free. 
In particular, $\cS(\bev(q,n)) = \bev(q,n+1)$ and $\cS(\prj{j_{t_n}}) = \prj{j_{t_{n+1}}}$.
Let us look closer to $\wt_i(q, S)$.
The operator $\wt_i(q, S)$ corresponds to the circuit depicted in \autoref{fig:whileRec}. Specifically, it includes the component $\bev(q, 1)$, which represents the first controlled-not operation, and it incorporates the operator defined by $(\prj{0_{t_1}} + \prj{1_{t_1}} \cdot \cS(\wt_{n-1}(q, S)) S)$. 
This expression represents the operation $S$ followed by $\cS(\wt_{n-1}(q, S))$ controlled by $t_1$. 
Here, $\cS(\wt_{n-1}(q, S)) S$ denotes the composition of the unitary $S$, which encodes the semantics of the body of the while loop, and $\cS(\wt_{n-1}(q, S))$, which captures the semantics of the remaining part of the while loop.
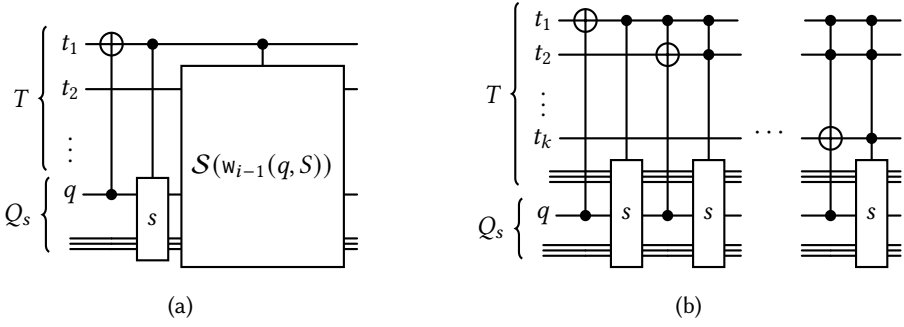
\begin{figure}
    \centering
    \begin{subfigure}[b]{.48\linewidth}
        \centering
        \begin{quantikz}[wire types={q,q,n,q,b},classical gap=0.07cm,row sep=0.35em, column sep = 0.5em]
            \lstick[3]{$T$} \ t_{1} \ & \targ{} & \ctrl{3} & \ctrl{1} &\\
                \ t_{2} \ & & & \gate[4]{\cS(\wt_{i-1}(q,S))} &\\
                \ \vdots \ & & &\\[.25em]
            \lstick[2]{$Q_s$} \ q \ & \ctrl{-3} & \gate[2]{s} & &\\
                & & & &
        \end{quantikz}
        \caption{}\label{fig:whileRec}
    \end{subfigure}
    \begin{subfigure}{.48\linewidth}
        \centering
        \begin{quantikz}[wire types={q,q,n,q,b,q,b},classical gap=0.07cm,row sep=0.35em, column sep = 0.5em]
            \lstick[5]{$T$} \ t_{1} \ & \targ{} & \ctrl{4} & \ctrl{1} & \ctrl{1} & \midstick[7,brackets=none]{$\cdots$} & \ctrl{1} & \ctrl{1} & \\
                \ t_{2} \ & & & \targ{} & \ctrl{4} & & \ctrl{2} & \ctrl{2} & \\
                \ \vdots \ & & & & & & & & \\
                \ t_{k} \ & & & & & & \targ{} & \ctrl{2} & \\
                \  \ & & \gate[3]{s} & & \gate[3]{s} & & & \gate[3]{s} & \\ [.15em]
            \lstick[2]{$Q_s$} \ q \ & \ctrl{-5} & & \ctrl{-4} & & & \ctrl{-2} & &  \\
                & & & & & & & &
        \end{quantikz}
        \caption{}
        \label{fig:whileBcirc}
    \end{subfigure}
    \caption{Finite $k$-while loop}
\end{figure}

\begin{restatable}{proposition}{propunind}
The closed formula of \autoref{eq: wnricdef}, is:
\begin{align}\label{eq: Wnclose}
    \Wt_n(q, S) = \sum_{h=1}^{n}(\prod_{i=1}^{h-1}(\prj{1_{t_i}})\! \cdot\! \prj{0_{t_h}} \! \cdot\! \prod_{i=n-h}^{n-2}(\bev(q,n-i)\,S) \cdot \bev(q,1)) + \prod_{i=1}^{n}\prj{1_{t_i}} \!\cdot\! \prod_{i=0}^{n-1}S\bev(q,n-i).
\end{align}
\end{restatable}
The proof is given in \autoref{appx}.

Given a $\whle{q}{s}$ statement, we can build the chain of finite unitary approximation $\{\Wt(q,\unSem{s}\}_n$.
To define the semantics of the general construct, we must now consider the case of an infinite loop.
However, we can show that there is no limit to this sequence.
From~\cite{vanini_analysis_nodate}, we recall the notion of \emph{strong convergence} and an important theorem about the convergence of an infinite sequence of operators.
\begin{definition}\label{def: strongconv}
    Let $T_n$ and $T$ be linear operators from $\Dom$ to itself.
    If $\|T_n \ket{\psi} - T \ket{\psi}\| \to 0$ as $n \to \infty$, $\forall \ket{\psi} \in \Dom$, then the sequence of operators $\{T_n\}$ is \textbf{strongly convergent} to $T$ (denoted as $T_n \to T$).
\end{definition}
\begin{theorem}\label{th: strongconv}
    Let $\{T_n\}_n$ be a sequence of bounded linear operators from $X \to Y$, where $X$ and $Y$ are Banach spaces.
    $T_n \to T$ if and only if the sequence $\{\|T_n\|\}_n$ is bounded and the sequence $\{T_n x\}_n$ is a Cauchy sequence in $Y$ for all $x \in M \subset X$, where the span of $M$ is dense in $X$.
\end{theorem}
We recall that a sequence $\{x_n\}_n$ in a Hilbert space is said to be a Cauchy sequence if, for any positive real number $\epsilon > 0$, there exists a positive integer $N$ such that for all positive integers $m$ and $n$ greater than $N$, the distance $\|x_m - x_n\| < \epsilon$.

As an example, consider the quantum program $\whle{q}{\skp}$, evaluated on $\ket{\psi}_P = \ket{0\dots}_T\ket{1}_q$.
This loop corresponds to the following chain:
\begin{equation}\label{eq: esUni3}
    \begin{aligned}
    \Wt_0(q, X_q)\ket{0\dots}_T\ket{1}_q &= \ket{0\dots}_T\ket{1}_q, \\
    \Wt_1(q, X_q)\ket{0\dots}_T\ket{1}_q &= \ket{10\dots}_T\ket{1}_q, \\
    % \Wt_2(q, X_q)\ket{0\dots}_T\ket{1}_q &= \ket{110\dots}_T\ket{1}_q, \\
    \Wt_n(q, X_q)\ket{0\dots}_T\ket{1}_q &= \ket{1^{\otimes n}0\dots}_T\ket{1}_q.
\end{aligned}
\end{equation}
We see that for all $n$, $\|\Wt_{n+1}(q, X_q)\ket{\psi}_P - \Wt_n(q, X_q)\ket{\psi}_P\| = 2$. More in general, if 
$$
\ket{\psi}_P = \ket{0\dots}_T(\alpha\ket{0} + \beta\ket{1})_q,
$$
we have 
$$
\Wt_n(q, X_q)\ket{\psi}_P = \alpha\ket{0\dots}_T\ket{0}_q + \beta\ket{1^{\otimes n}0\dots}_T\ket{1}_q,
$$
and for all $n$, $\|\Wt_{n+1}(q, X_q)\ket{\psi}_P - \Wt_n(q, X_q)\ket{\psi}_P\| = 2\beta^2$.
Since for infinite vectors the distance between the element of $\{\Wt_n(q, X_q)\ket{\psi}_P\}_n$ is constant, the sequence $\{\Wt_n(q, X_q)\ket{\psi}_P\}_n$ is not Cauchy, and for \autoref{th: strongconv} the sequence $\{\Wt_n(q, X_q)\}_n$ is not strongly convergent.
The unitary semantics is therefore only able to model a while statement limited to a certain finite number of iterations $k$, which we will denote by $\whleB{q}{s}$.
The semantics for this statement is defined as follows:
$$
\unSem{\whleB{q}{s}} = \Wt_k(q,\unSem{s}).
$$
This unitary operator can be represented by the circuit in \autoref{fig:whileBcirc}.

\subsection{Examples}\label{sec: unsem ex}
We show how the operator $\Wt_n$ works by means of some examples.

First we consider $\whleB{q}{X(q)}$ with the state $\ket{\psi} = \sfrac{1}{\sqrt{2}}\ket{0\dots}_{T}(\ket{0}_q+\ket{1}_q))$ as input:
\begin{equation}\label{eq: esUni1}
    \begin{aligned}
    \Wt_0(q,X_q)\ket{\psi} &= \sfrac{1}{\sqrt{2}}\ket{0\dots}_{T}(\ket{0}_q+\ket{1}_q)\\
    \Wt_1(q,X_q)\ket{\psi} &= \sfrac{1}{\sqrt{2}}(\ket{0\dots}_{T}\ket{0}_q + \ket{10\dots}_{T}\ket{0}_q)\\
    \Wt_2(q,X_q)\ket{\psi} &= \sfrac{1}{\sqrt{2}}(\ket{0\dots}_{T}\ket{0}_q + \ket{10\dots}_{T}\ket{0}_q)\\
    \dots
    \end{aligned}
\end{equation}
Note that with this input, the while loop fully terminates, and in fact, $\Wt_n(q, X_q)$ reaches a fixpoint.

Things are different when non-termination is involved, as in the loop $\whle{q}{H(q)}$ with the same input $\ket{\psi}$:
\begin{equation}\label{eq: esUni2}
    \begin{aligned}
    \Wt_0(q,H_q)\ket{\psi} &= \sfrac{1}{\sqrt{2}}\ket{0\dots}_{T}(\ket{0}_q+\ket{1}_q)\\
    \Wt_1(q,H_q)\ket{\psi} &= \sfrac{1}{\sqrt{2}}\ket{0\dots}_{T}\ket{0}_q + \sfrac{1}{2}(\ket{10\dots}_{T}\ket{0}_q - \ket{10\dots}_{T}\ket{1}_q)\\
    \Wt_2(q,H_q)\ket{\psi} &= \sfrac{1}{\sqrt{2}}\ket{0\dots}_{T}\ket{0}_q + \sfrac{1}{2}\ket{10\dots}_{T}\ket{0}_q \\
    &+ \sfrac{1}{\sqrt{8}}(\ket{110\dots}_{T}\ket{0}_q - \ket{110\dots}_{T}\ket{1}_q)\\
    \Wt_3(q,H_q)\ket{\psi} &= \sfrac{1}{\sqrt{2}}\ket{0\dots}_{T}\ket{0}_q + \sfrac{1}{2}\ket{10\dots}_{T}\ket{0}_q \\
    &+ \sfrac{1}{\sqrt{8}}\ket{110\dots}_{T}\ket{0}_q +
    \sfrac{1}{4}(\ket{1110\dots}_{T}\ket{0}_q - \ket{1110\dots}_{T}\ket{1}_q)\\
    \Wt_n(q,H_q)\ket{\psi} &= \Sigma_{i=0}^{n}\sfrac{1}{\sqrt{2^{i+1}}} \ket{1^{\otimes i}0\dots}_{T}\ket{0}_q - 
    \sfrac{1}{\sqrt{2^{n+1}}} \ket{1^{\otimes n}0\dots}_{T}\ket{1}_q.
    \end{aligned}
\end{equation}  
Since each $t_j$ controls the execution of the $j$-th iteration, $t_j = 1$ indicates that the $j$-th iteration has been executed (see \autoref{fig:whileBcirc} and \autoref{fig:whileCirc}). 
During each iteration of $\wt_n$, the terminating part of the state—where $q$ equals $\ket{0}$—is gradually increased. In contrast, the portion corresponding to non-termination, where $q$ equals $\ket{1}$, decreases but never fully reaches zero.

Finally, consider again the while loop $\whle{q}{\skp}$, evaluated on $\ket{0\dots}_T\ket{1}_q$ in \autoref{eq: esUni3}.
Here for a divergent loop, each $\wt_n$ differs from the previous ones, reflecting that the unitary semantics corresponds to the partial results of the divergent loop computation.

\section{Linear Semantics}\label{sec: lnSem}
To model infinite computation in quantum computing, we need to enlarge the domain of denotations so as to include an appropriate limit.
To this purpose, we observe that a unitary operator is also bounded and, therefore, can be seen as an element of the Banach Space of linear bounded operators on $\Dom$.

We start by defining the linear semantics for each statement in $SL$, except the \texttt{while} statement.
\begin{definition}
    Let $\Bnd(\Dom)$ be the space of linear bounded operators from $\Dom$ to $\Dom$ equipped with the operator norm $\|A\| = \sup_{\|\ket{\psi}\| = 1} \|A\ket{\psi}\|$.
    The linear semantics is a function $\lnSem{\cdot}: s \to \Bnd(\Dom)$, defined by:
    \begin{enumerate}
        \item $\lnSem{\skp} = \Ident_{\Dom}$;
        \item $\lnSem{U(\arrw{q})} = \arrw{U}$, where $\arrw{U} = \Ident \otimes U \otimes \Ident$, i.e. the extension of $U$ on $\Dom$;
        \item $\lnSem{s_1;\, s_2} = \lnSem{s_2}\cdot\lnSem{s_1}$;
\end{enumerate}
\end{definition}
Note that the linear semantics of these statements is exactly the same as the unitary semantics defined in \autoref{def:unSem}.
For the \texttt{while} statement, we need instead to introduce a new bounded linear operator, which will allow us to give a meaning also to infinite computations.

Let $S$ be a bounded linear operator and $q \in Q_s$. We define the operator $\lt_n(q,S)$ by:
\begin{equation}\label{eq: lnricdef}
    \begin{aligned}
        \lt_0(q,S) &= \textit{0}\\
        \lt_n(q,S) &= (\prj{0_{t_1}} + \prj{1_{t_1}} \cdot \cS(\lt_{n-1}(q,S)) \cdot S) \bev(q,1),\\
    \end{aligned}
\end{equation}
where $\textit{0}$ is the zero operator in $\Dom$, $\cS$ produces the shift $t_n \rightarrow t_{n+1}$, and the controlled operation is represented by the operator defined in \autoref{prop: ctrl}.

It is easy to see that for $n>0$, the operator $\lt_n(q, S)$ is equivalent to the operator $\wt_n(q, S)$ as defined in \autoref{eq: wnricdef}; in fact, $\lt_n(q, S)$ is defined as the composition of the controlled-not gate, which evaluates the guard, and the operator corresponding to the composition of the semantics of the loop body $S$ and $\lt_{n-1}(q, S)$ shifted by $\cS$.

\begin{proposition}\label{prop: boundln}
    For all $n$, if $S$ is bounded then $\lt_n(q,S)$ is bounded. 
\end{proposition}
\begin{proof}
    Let's consider $\lt_n(q,S)$ and a vector $\ket{\psi}$ such that $\| \psi \| = 1$:
    \begin{align*}
        \| (\prj{0_{t_1}} + \prj{1_{t_1}} \cdot \cS(\lt_{i-1}(q,S)) \cdot S)\bev(q,1)\ket{\psi} \|^2 = \\
        &\text{(since }\bev(q,1) \text{ is unitary)}\\
        = \| (\prj{0_{t_1}} + \prj{1_{t_1}} \cdot \cS(\lt_{i-1}(q,S)) \cdot S)\ket{\psi} \|^2 = \\ 
         = \| \prj{0_{t_1}}\ket{\psi} + \prj{1_{t_1}} \cdot \cS(\lt_{i-1}(q,S)) \cdot S\ket{\psi} \|^2 =\\ 
        \hspace{10em}
        &\text{($\prj{0_{t_1}}\!\!\ket{\psi}$ and $\prj{1_{t_1}}\!\cS(\lt_{i-1}(q,S)) S\!\ket{\psi}$ are orthogonal \cite{heinosaari2008guidemathematical})}\\
        = \| \prj{0_{t_1}}\ket{\psi} \|^2 + \| \prj{1_{t_1}} \cS(\lt_{i-1}(q,S)) S)\ket{\psi} \|^2 \le \\
        &(\cS(\lt_{i-1}(q,S)) S) \text{ is bounded)}\\
        \le \| \prj{0_{t_1}}\ket{\psi} \|^2 + \| \prj{1_{t_1}} \ket{\psi} \|^2 =\\
        = \| (\prj{0_{t_1}} + \prj{1_{t_1}})\ket{\psi} \|^2 =\| \ket{\psi} \|^2 = 1.
    \end{align*}
    Since $\|\lt_n(q,S) \ket{\psi} \|^2 \le 1$ also $\|\lt_n(q,S) \ket{\psi} \| \le 1$.
\end{proof}

\begin{restatable}{proposition}{proplnind}
The closed formula of \autoref{eq: lnricdef} is:
\begin{align}\label{eq: Lnclose}
    \Lt_n(q,S) = \sum_{k=1}^{n}(\prod_{i=1}^{k-1}(\prj{1_{t_i}}) \cdot \prj{0_{t_k}} \cdot \prod_{i=n-k}^{n-2}(\bev(q,n-i)\cdot S) \cdot \bev(q,1)).
\end{align}
\end{restatable}
A proof by induction is given in \autoref{appx}.

Now, given the $\whle{q}{s}$ statement, we can define a chain of linear operator $\{\Lt_n(q,\lnSem{s})\}_n$ that represent the linear semantics of all finite executions of the while loop.
\begin{proposition}\label{prop: boundlnsem}
    For all $n$, $\Lt_n(q,\lnSem{s})$ is bounded.
\end{proposition}
\begin{proof}
    The proof follows by means of a structural induction on $s$ and \autoref{prop: boundln}
\end{proof}
To define the semantics of the whole construct, we now must include the case of an infinite loop, which, as we will show later, is captured by the limit of the sequence $\{\Lt_n(q,\lnSem{s})\}_n$.

From \autoref{th: strongconv}, we know that if we prove that the sequence $\{\|\Lt_n(q, \lnSem{s})\|\}$ is bounded and that, for all $\ket{\psi} \in \Dom$, the sequence $\{\Lt_n(q, \lnSem{s}) \ket{\psi}\}$ is a Cauchy sequence in $\Dom$, then we can conclude that the sequence $\{\Lt_n(q,\lnSem{s})\}_n$ has a limit.
So, let's first prove that $\{\|\Lt_n(q,\lnSem{s})\|\}_n$ is bounded.

\begin{proposition}\label{prop: lnseqbound}
    The sequence $\{\|\Lt_n(q, \lnSem{s})\|\}$ is bounded.
\end{proposition}

\begin{proof}
    From \autoref{prop: boundlnsem}, we know that for all $n$, $\|\Lt_n(q, \lnSem{s})\| \le 1$, thus the sequence $\{\|\Lt_n(q, \lnSem{s})\|\}$ is bounded above~\cite{mattuck2013introduction}.
\end{proof}

We will now show that, for all $\ket{\psi} \in \Dom$, the sequence $\{\Lt_n(q, \lnSem{s}) \ket{\psi}\}$ is Cauchy. To this purpose, we state some support lemmas.

\begin{lemma}\label{prop: diffVect}
    For all $n > 0$, $\Lt_n(q, \lnSem{s}) = \sum_{i=1}^n \left( \Lt_i(q, \lnSem{s}) - \Lt_{i-1}(q, \lnSem{s}) \right)$.
\end{lemma}

\begin{proof}
    In general, given a sequence $A_n$, by induction on $n$ and by straightforward arithmetic simplifications, we can prove that $A_n = \sum_{i=1}^n (A_i - A_{i-1}) + A_0$.  
    In our case, $\Lt_0(q, \lnSem{s}) = 0$, thus $\Lt_n(q, \lnSem{s}) = \sum_{i=1}^n \left( \Lt_i(q, \lnSem{s}) - \Lt_{i-1}(q, \lnSem{s}) \right)$.
    % For instance, $\Lt_2(q, \lnSem{s})\!=\! \sum_{i=1}^2\!\left( \Lt_i(q, \lnSem{s}) - \Lt_{i-1}(q, \lnSem{s}) \right) = \Lt_2(q, \lnSem{s}) - \Lt_1(q, \lnSem{s}) + \Lt_1(q, \lnSem{s})$.
\end{proof}

\begin{lemma}\label{prop: orthoDiff}
    For every $n \neq m$, $\Lt_n(q, \lnSem{s}) - \Lt_{n-1}(q, \lnSem{s})$ and $\Lt_m(q, \lnSem{s}) - \Lt_{m-1}(q, \lnSem{s})$ are orthogonal.
\end{lemma}

\begin{proof}
    From \autoref{eq: Lnclose}, we can compute:
    \begin{equation*}
        \begin{gathered}
        \Lt_n(q, \lnSem{s}) - \Lt_{n-1}(q, \lnSem{s}) = \prod_{i=1}^{n-1} \left( \prj{1_{t_i}} \right) \cdot \prj{0_{t_n}} \cdot \prod_{i=0}^{n-2} \left( \bev(q, n - i) \cdot S \right) \cdot \bev(q, 1), \\
        \Lt_m(q, \lnSem{s}) - \Lt_{m-1}(q, \lnSem{s}) = \prod_{i=1}^{m-1} \left( \prj{1_{t_i}} \right) \cdot \prj{0_{t_m}} \cdot \prod_{i=0}^{m-2} \left( \bev(q, n - i) \cdot S \right) \cdot \bev(q, 1).
        \end{gathered}
    \end{equation*}
    Clearly, if $n \neq m$, $\prod_{i=1}^{n-1} \left( \prj{1_{t_i}} \right) \prj{0_{t_n}}$ and $\prod_{i=1}^{m-1} \left( \prj{1_{t_i}} \right) \prj{0_{t_m}}$ are orthogonal, and thus $\Lt_n(q, \lnSem{s}) - \Lt_{n-1}(q, \lnSem{s})$ and $\Lt_m(q, \lnSem{s}) - \Lt_{m-1}(q, \lnSem{s})$ are orthogonal.
\end{proof}

\begin{lemma}\label{prop: increasing}
    For every $n, m$, if $n \leq m$, then $\|\Lt_n(q, \lnSem{s}) \ket{\psi}\| \leq \|\Lt_m(q, \lnSem{s}) \ket{\psi}\|$.
\end{lemma}

\begin{proof}
    From \autoref{eq: Lnclose}, we observe that if $n \leq m$, the summation corresponding to $\Lt_m(q, \lnSem{s})$ contains all the terms of $\Lt_n(q, \lnSem{s})$ plus additional terms.  
    Recall that each term in the summation consists of projectors that are mutually orthogonal, and the sum of all projectors is different from the identity.  
    Thus, if $\Lt_n(q, \lnSem{s}) \ket{\psi} = \sum \alpha_i \ket{e_i}$ and $\Lt_m(q, \lnSem{s}) \ket{\psi} = \sum \beta_i \ket{e_i}$, where $\ket{e_i}$ is a standard basis of $\Dom$, then $\forall i, \, \alpha_i \neq 0 \Rightarrow \beta_i = \alpha_i$.
    In other words, $\Lt_n(q, \lnSem{s}) \ket{\psi}$ is a substate of $\Lt_m(q, \lnSem{s}) \ket{\psi}$ since $\Lt_m(q, \lnSem{s})$ projects the state on `more basis vectors' that $\Lt_n(q, \lnSem{s})$.
\end{proof}

\begin{theorem}
    For all vectors $\ket{\psi} \in \Dom$, the sequence $\{\Lt_n(q, \lnSem{s}) \ket{\psi}\}_n$ is Cauchy.
\end{theorem}
\begin{proof}
    By \autoref{prop: diffVect}, $\Lt_n(q,\lnSem{s})\ket{\psi} = \Sigma_{i=1}^n (\Lt_i(q,\lnSem{s}) - \Lt_{i-1}(q,\lnSem{s}))\ket{\psi}$.
    By \autoref{prop: orthoDiff} and \cite[Exercise 2]{heinosaari2008guidemathematical} it holds that:
    $$
    \|\Lt_n(q,\lnSem{s})\ket{\psi}\|^2 = \Sigma_{i=1}^n \|( \Lt_i(q,\lnSem{s}) - \Lt_{i-1}(q,\lnSem{s}))\ket{\psi} \|^2.
    $$
    Let us consider the sequence of partial sums $\{\Sigma_{j=1}^n \|( \Lt_i(q,\lnSem{s}) - \Lt_{i-1}(q,\lnSem{s}))\ket{\psi} \|^2 \}_n$.  
    By \autoref{prop: boundlnsem} and \autoref{prop: increasing}, we know that all partial sums are bounded and that the sequence of partial sums is increasing.   
    Since the sequence of partial sums is both increasing and bounded above, it must converge~\cite{mattuck2013introduction}.
    This implies that the infinite sum $\Sigma_{j=1}^{\infty} \|( \Lt_i(q,\lnSem{s}) - \Lt_{i-1}(q,\lnSem{s}))\ket{\psi} \|^2$ also converges~\cite{spivak08}.
    We have that $\Sigma_{i=1}^{\infty} \| ( \Lt_i(q, \lnSem{s}) - \Lt_{i-1}(q, \lnSem{s}) ) \ket{\psi} \|^2$  
    converges, and for all $i$,  
    $ \| ( \Lt_i(q, \lnSem{s}) - \Lt_{i-1}(q, \lnSem{s}) ) \ket{\psi} \|^2 > 0$.  
    Thus, $\lim_{i \rightarrow \infty} \left( \| ( \Lt_i(q, \lnSem{s}) - \Lt_{i-1}(q, \lnSem{s}) ) \ket{\psi} \|^2 \right) = 0$, which implies that $\lim_{i \rightarrow \infty} \left( \| \Lt_i(q, \lnSem{s}) \ket{\psi} - \Lt_{i-1}(q, \lnSem{s}) \ket{\psi} \| \right) = 0$.  
    This ensures that the sequence $\{ \Lt_i(q, \lnSem{s}) \ket{\psi} \}_i$ is a Cauchy sequence in $\Dom$.
\end{proof}

Having shown that, for all $\ket{\psi} \in \Dom$, 
$\{\Lt_n(q,\lnSem{s} \ket{\psi}\}_n$ is a Cauchy sequence,
by \autoref{th: strongconv} we can conclude that the sequence $\{\Lt_n(q,\lnSem{s})\}$ has a limit.
Thus, we use that limit to define the semantics of the quantum loop:
$$
\lnSem{\whle{q}{s}} = \lim_{n \to \infty} \Lt_n(q, \lnSem{s}).
$$
Moreover, the following proposition shows that this limit corresponds to a well-defined linear operator on $\Dom$.
\begin{proposition}
    Let $\{\ket{e_i}\}_i$ be the set of standard basis of $\Dom$, and let $\cL(q,\lnSem{s})$ be the linear operator defined as $\cL(q,\lnSem{s}) \ket{e_i} = \lim_{n \to \infty}\Lt_n(q,\lnSem{s}) \ket{e_i}$, $\forall \ket{e_i}$. Then 
    \[
    \cL(q,\lnSem{s})= \lim_{n \to \infty}\{\Lt_n(q,\lnSem{s})\}_n
    \].
\end{proposition}
\begin{proof}
    By definition, for all basis $\ket{e_i}$ of $\Dom$, $\cL(q,\lnSem{s}) \ket{e_i} = \lim_{n \to \infty} \Lt_n(q,\lnSem{s}) \ket{e_i}$, thus $\|\Lt_n(q,\lnSem{s}) \ket{e_i} - \cL(q,\lnSem{s}) \ket{e_i}\| \to 0$ as $n \to \infty$.
    Therefore, for all $\ket{\psi} \in \Dom$, $\|\Lt_n(q,\lnSem{s}) \ket{\psi} - \cL(q,\lnSem{s}) \ket{\psi}\| \to 0$ as $n \to \infty$.    
\end{proof}

\section{Relation between Unitary and Linear Semantics}
So far, we have introduced two semantics:
\begin{description}
    \item[-] The unitary semantics models exactly the behavior of quantum computation, and
    \item[-] the linear semantics allows us to define a fixpoint and define the semantics of the infinite loops.
\end{description}

In this section, we investigate the relationship between these two semantics.

Let us consider the examples introduced in \autoref{sec: unsem ex}.
We can construct the semantics of the program $\whle{q}{X(q)}$, starting from the state $\ket{\psi} = \sfrac{1}{\sqrt{2}}\ket{0\dots}_{T}(\ket{0}_q+\ket{1}_q)$, as follows:
% Let's consider the various $\Lt_n(q,X_q)$:
\begin{equation}\label{eq: esLni1}
    \begin{aligned}
    \Lt_0(q,X_q)\ket{\psi} &= 0\\
    \Lt_1(q,X_q)\ket{\psi} &= \sfrac{1}{\sqrt{2}}\ket{0\dots}_{T}\ket{0}_q\\
    \Lt_2(q,X_q)\ket{\psi} &= \sfrac{1}{\sqrt{2}}(\ket{0\dots}_{T}\ket{0}_q + \ket{10\dots}_{T}\ket{0}_q)\\
    \Lt_3(q,X_q)\ket{\psi} &= \sfrac{1}{\sqrt{2}}(\ket{0\dots}_{T}\ket{0}_q + \ket{10\dots}_{T}\ket{0}_q)
    \end{aligned}
\end{equation}    
We see that we have 'collected' the sub-state that corresponds to the terminating execution, and since the program is fully terminating, the fixpoint is a state with a norm equal to $1$.
Comparing the examples in the previous sections (\autoref{eq: esUni1}), it is easy to see that when the loop terminates, the two semantics coincide, specifically $\Lt_{n+1} = \Wt_n$. 
In this case, the linear semantics is `behind' the unitary semantics because the latter also includes the component that corresponds to executions that `keep going.'

Now consider the program $\whle{q}{H(q)}$ and the construction of its semantics:
\begin{equation}\label{eq: esLni2}
    \begin{aligned}
        \Lt_0(q,H_q)\ket{\psi} &= 0\\
        \Lt_1(q,H_q)\ket{\psi} &= \sfrac{1}{\sqrt{2}}\ket{0\dots}_{T}\ket{0}_q\\
        \Lt_2(q,H_q)\ket{\psi} &= \sfrac{1}{\sqrt{2}}\ket{0\dots}_{T}\ket{0}_q + \sfrac{1}{2}\ket{10\dots}_{T}\ket{0}_q\\
        \Lt_3(q,H_q)\ket{\psi} &= \sfrac{1}{\sqrt{2}}\ket{0\dots}_{T}\ket{0}_q + \sfrac{1}{2}\ket{10\dots}_{T}\ket{0}_q + \sfrac{1}{\sqrt{8}} \ket{110\dots}_{T}\ket{0}_q)\\
        \Lt_n(q,H_q)\ket{\psi} &= \sum_{i=1}^{n}\sfrac{1}{\sqrt{2^n}} \ket{1^{\otimes i}0\dots}_{T}\ket{0}_q.
    \end{aligned}    
\end{equation}
Here, we see that in each approximation $\Lt_n$, we obtain a state with a norm less than 1, which corresponds to the part of the execution that terminates in at most $n-1$ iterations.
If we compare the linear semantics in \autoref{eq: esLni2} with the corresponding unitary semantics in \autoref{eq: esUni2}, we observe that the linear semantics has a `missing' part in the output states—specifically, the portion of the state representing the execution that `keeps going'.

Finally, consider the while loop $\whle{q}{skip}$, evaluated on $\ket{0\dots}_{T}\ket{1}_q$:
\begin{equation}\label{eq: esLni3}
    \begin{aligned}
    \Lt_0(q,\Ident)\ket{0\dots}_{T}\ket{1}_q &= 0\\
    \Lt_1(q,\Ident)\ket{0\dots}_{T}\ket{1}_q &= \prj{0_{t_1}}\ket{10\dots}_{T}\ket{1}_q = 0\\
    \Lt_2(q,\Ident)\ket{0\dots}_{T}\ket{1}_q &= \prj{0_{t_1}}\ket{10\dots}_{T}\ket{1}_q + \prj{1_{t_1}}\prj{0_{t_2}}\ket{110\dots}_{T}\ket{1}_q = 0.\\
    \end{aligned}
\end{equation}
When dealing with a loop that diverges for the whole quantum state, the linear semantics in \autoref{eq: esLni3} evaluates to $0$.
%as the linear semantics of a non-terminating loop is the zero operator. 
On the other hand, the unitary semantics (\autoref{eq: esUni3}) computes all partial executions of the divergent loop.
% We see that the unitary semantics compute a partial evaluation of the at every step.

More generally, in all examples, the unitary operator $\Wt_n$ returns the portion of the state corresponding to executions that terminate after less than $n$ iterations, along with the partial results of computations that are still ongoing. In contrast, the linear semantics can separate the terminating portion of the computation.
In fact, by examining the definition of $\Wt_n$ (\autoref{eq: Wnclose}), we can identify two key components:
$$
\prj{0_{t_h}} \cdot \prod_{i=n-h}^{n-2} \big( \bev(q, n-i) S \big) \cdot \bev(q, 1),
$$
which is also present in the linear semantics (\autoref{eq: Lnclose}), and the final term is given by:
$$
\prod_{i=1}^n \prj{1_{t_i}} \cdot \prod_{i=0}^{n-1} S \bev(q, n-i).
$$
The first part corresponds to computations where the guard becomes $0$ after the $(n-1)$-th execution, indicating the terminating of the. 
The second part represents the non-terminating portion, where the guard remains true, meaning that the loop is still in progress.
The linear semantics, therefore, returns a sub-state of the result of the unitary semantics.  
For this reason, we can consider the linear semantics as an \textit{under-approximation} of the unitary semantics.  
By discarding part of the result, the linear semantics allows the definition of a limit and, therefore, gives meaning to infinite behaviors.

\section{Conclusion and Related Work}
% A first work on the definition of a denotational semantics of quantum languages is \cite{kashefi_quantum_2003}, where Kashefi introduces a domain theory in Quantum Computation.
We have introduced a denotational semantics for quantum programs, which approximates the unitary behaviour of the programs by means of linear operators acting on possibly non-normalized states, which contain both finite and infinite results. 
%In the following, we discuss some of the works that are related to ours.

Various approaches to the problem of modeling the control flow in a quantum progragram have been introduced in the literature on the design and implementation of quantum programming languages.
They can be grouped as follows.

%%%% Rewrite

\paragraph*{Probabilistic control flow}
The initial works on quantum program semantics focused on languages with quantum data and classical probabilistic control, primarily based on measurement operators.  
In~\cite{selinger2004towards}, Selinger introduced the basic notations, theories, and conventions for a quantum programming language with measurement-based probabilistic control flow, called QPL.  
He provided a denotational semantics for QPL by associating each program with a superoperator (a completely positive map that is not necessarily trace-preserving) in a finite-dimensional Hilbert space, represented as a morphism in a CPO-enriched traced monoidal category. 

In~\cite{perdrix2008a}, Perdrix extended this work by introducing a complete partial order (CPO) over admissible transformations, i.e., multisets of linear operators, and using it to define a denotational semantics for a simple quantum imperative language similar to QPL.  
He demonstrated that this semantics is an exact abstraction of Selinger's semantics.  

In a series of works~\cite{FENG2007151, ying2010quantumloop, ying2012defining, ying2012floyd}, Ying et al. explored a quantum while language with measurement-based probabilistic control flow.  
They defined a denotational semantics in terms of maps between density operators, generalizing Selinger's results to infinite-dimensional Hilbert spaces.  
In these works, the denotational semantics is constructed using the CPO of superoperators acting on partial density operators.  
Finally, in~\cite[Chapter 3]{MingFoundations}, Ying further developed this domain to define the semantics of a quantum language with recursion and measurement-based probabilistic control flow, providing a more comprehensive framework for reasoning about quantum programs.

\paragraph*{Quantum Control}
All the previous works focused on language with probabilistic control flow, thus their semantics is probabilistic and no superposition is introduced between possible executions of the programs.

A first form of quantum control was introduced in Altenkirch and Grattage's functional language QML~\cite{altenkirch2005functional}, a first-order functional language on finite types equipped with a categorical semantics capturing only finite quantum computations.
In \cite{lampis2008quantum}, Lampis et al. introduced nQML, a simplified version of QML with simpler control constructs and a denotational semantics based on density matrices and unitary transformations, still capturing only finite computations.

A formal definition of quantum imperative language with quantum control flow was introduced by Ying et al. in~\cite{ying2012defining,ying2014alternation}, where they define the QuGCL language, i.e. a language with both quantum control flow and measurement-based control flow but no recursion. A semantics for this language is given in terms of a new mathematical tool called the guarded composition of operator-valued functions, where operator-valued functions are defined using the Kraus operator-sum representation~\cite[Chapter 8]{MichealANielsen}.
However, this semantics is not compositional.
%(only semi-compositional).

More recently, in ~\cite{valiron2022semantics}, Valiron has reviewed the use of quantum control in the $\lambda$-calculus, explaining the difference between superposition of terms and superposition of data in the various formulations.
%??????

In~\cite{yuan2024controlmachine}, Yuan et al. studied the problem of automatically compiling self-controlled quantum operations minimizing the use of extra temporary variables and avoiding that these extra temporary variables are entangled with the rest of the variables.
In particular, they formalize the conditions under which this compilation is possible.
Their compilation technique can be seen as a particular case of our unitary semantics.
Inspired by Yuan et al., Zhang and Ying~\cite{zhang2024quantumregister} propose a quantum architecture that supports quantum control flow and finite quantum recursion.

\paragraph*{Quantum recursion}
Introducing a construct for quantum control flow in a programming language naturally leads to the concepts of quantum recursion or quantum loops.  
An initial idea of quantum recursion, based on using an infinite set of external coins, was informally discussed in~\cite{ying2012defining}.  

The issue of quantum recursion was formally addressed for the first time by Badescu and Panangaden in~\cite{badescu2015quantum}.  
They extended the QPL programming language~\cite{selinger2004towards} by introducing a quantum if-statement and provided a denotational semantics for this extension based on Kraus decomposition.  
However, they observed that the semantics of quantum case statements is not monotone with respect to Selinger's order, concluding that the existing framework is inadequate for modeling quantum recursion.  
In our paper, we generalise this observation by demonstrating that even in a setting without measurement, it is impossible to define a limit for the sequence of unitary operators. 
This highlights the inherent challenges in modeling quantum recursion within a purely unitary framework. 

In~\cite{sabry2018symmetric},  Sabry et al. extend a classical, typed, reversible language that includes lists and fixpoints to a quantum setting.
%The expanded language is purely quantum, without measurement, allows for linear combinations of terms, and restricts fixpoints to structurally recursive fixpoints. This restriction allows for the representation of quantum programs as unitary operators in infinite-dimensional Hilbert spaces. 
The resulting quantum language is provided with an operational semantics following the algebraic $\lambda$-calculi principles.
This work proves that it is possible to represent with a unitary operator a quantum program with recursion, only if we are able to construct the fixpoint of the recursive call by means of a finite unfolding the recursive calls.
As we do not impose any restriction on the quantum loop, our semantics is more general.

%Andrés-Martínez, M. et al. \cite{andres2022weakly} model quantum loop using weak measurements.  A weak measure is a measurement that gives us partial information about the system on average, but also disturbs the state as little as possible \cite{brun2002simple}.
%In this way, they are able to decide when to stop a while loop, limiting the perturbation caused by a measurement at each iteration.
In his PhD thesis \cite{andres2022unbounded}, Andrés-Martínez introduces a quantum while language similar to ours but equipped with a categorical semantics.
The thesis extends Haghverdi’s unique decomposition categories—originally introduced to model iteration in classical computation—by addressing their incompatibility with the quantum settings. 
% While Haghverdi’s framework relies on a notion of addition that lacks additive inverses, Andrés-Martínez generalises the additive structure to accommodate them, thus capturing destructive interference, a core aspect of quantum theory. 
This generalisation establishes connections to topological groups and leads to a hierarchy of categories enriched with infinitary addition and convergence criteria. 
Building on this foundation, the execution formula is shown to define a valid categorical trace even over categories of quantum processes on infinite-dimensional Hilbert spaces.
%Intuitively, in the thesis, they assume that input and output are interpreted in terms of plane waves, instead of in the form of particles, which are, in some sense, everywhere at once. 
%In this way, when we execute a quantum loop, it is like all the executions reach the fixed point instantly.{\color{red}PBS work}
%These approaches are indeed interesting. 
This approach, however, relies on a computational model which is not immediately referable to quantum circuits.
In defining our semantics, we, instead,  refer to the standard computation model of quantum computing; in fact, our unitary semantics is directly implementable on a quantum computer.

In~\cite{MingFoundations, ying2014recursion}, Ying explores a problem similar to ours, considering a recursive quantum language.
%In the following Section, we analyze this approach in detail.
This work has a stronger similarity with our approach than the other works mentioned above, although
%\subsection{Comparison with~\cite{MingFoundations, ying2014recursion}}
%In~\cite{MingFoundations, ying2014recursion}, 
the implementation of the self-controlled operation is in a sense `dual' with respect to our model. In fact,
instead of making a quantum copy of the guard using a \CN\ gate, Ying proposes to prepare and carry along the program an infinite number of copies of identical qubits to represent the guard.
%For example, the self-controlled operation $\qift{q}{U(q)}$, where the state of $q$ is $\alpha\ket{0}_q + \beta\ket{1}_q$, is performed using two identical qubits $q_0$ and $q_1$, along with an operator $CU$ defined as $CU(\alpha\ket{0}_{q_0} + \beta\ket{1}_{q_0}) \otimes (\alpha\ket{0}_{q_1} + \beta\ket{1}_{q_1}) =  \alpha\ket{0}_{q_0}(\alpha\ket{0}_{q_1} + \beta\ket{1}_{q_1}) + \beta\ket{1}_{q_0}U(\alpha\ket{0}_{q_1} + \beta\ket{1}_{q_1})$.
%Note that this result is slightly different from our approach. 
However, since we are trying to represent in a unitary way an operation that inherently cannot be unitary, this difference is merely a design choice to address the unitarity constraint.

\begin{figure}
    \centering
    \begin{quantikz}[wire types={q,q,q,b},row sep=.7em,classical gap=0.07cm]
        \lstick{$\ket{\psi}_{q_0}$} & \ctrl{3} & \ctrl{3} & \ctrl{3} &\\
        \lstick{$\ket{\psi}_{q_1}$} & \gate{U} & \ctrl{2} & \ctrl{2} &\\
        \lstick{$\ket{\psi}_{q_2}$} & \gate{U} & \gate{U} & \ctrl{1} &\\
        & \gate{\textbf{P}} & \gate{\textbf{P}} & \gate{\textbf{P}} &\\
    \end{quantikz}
    \caption{$\textbf{P} := \qift{q}{U(q);\textbf{P}}$ circuit}
    \label{fig:ying rec}
\end{figure}

In this setting, recursion is achieved with programs of the form $\textbf{P}:= \qift{q}{U(q);\textbf{P}}$, which can be visually represented in \autoref{fig:ying rec}.
It can be seen that each recursive call consumes a copy of the guard variable, and to achieve infinite recursion, we need an infinite number of copies of the same guard variables. 
For defining their semantics, Ying et al. employ a formalism from quantum physics related to multi-particle systems, specifically Fock spaces~\cite{fock} and second quantization~\cite{bosonfermion}. 
In particular, they consider free Fock spaces, i.e., Hilbert spaces that describe quantum states with a variable number of indistinguishable particles, constructed as $\oplus_{n=1}^\infty \Hilb^{\otimes n}$, i.e., the direct sum of tensor powers of the single-particle Hilbert space $\Hilb$.
% This formalism allows for the representation of a variable number of copies of the same qubit, where each element of the direct sum corresponds to a situation with exactly $n$ copies of qubits.
Using this formulation, Ying defines a Complete Partial Order (CPO), which orders the operators within the Fock space based on the number of guard copies that these operators act on.
As a result, the semantics of a recursive program demonstrates a monotonically continuous order, which, in turn, allows for the existence of a fixed point.

%The difference between our approach and Ying's is twofold. Firstly, we adopt a different method for performing self-controlled operations with a unitary operator. 
%The fundamental distinction lies in the formalism that Ying uses for the semantics of mathematical formulations. 
%His approach is capable of handling quantum systems that involve varying quantities of identical particles; moreover, he defines a denotational semantics through second quantization. 
%However, this formalism is based on Fock spaces, which are more abstract compared to the typical semantics used to model circuits.
%In contrast, we choose to work exclusively within the Hilbert space framework for the program variables. To manage infinite computations, we rely on linear operators, which, as we have demonstrated, are closely related to their corresponding unitary operations. 
%Additionally, our semantics specifically addresses the portion of execution that leads to termination.

Both ours and Ying's semantics share the key feature of utilising an infinite number of qubits to perform while loops and defining recursive unitary operators, as illustrated in \autoref{fig:whileCirc} and \autoref{fig:ying rec}.  
In fact, we can choose to formulate our approach within Ying’s Fock space semantics or, conversely, describe Ying’s semantics using our Unitary/Linear framework.  
In the latter case, we can specifically adopt the formulation introduced by Ying in \cite[Section 6.4]{ying2014recursion}, mapping the Fock space operator to a unitary operator on the program variables space as in \autoref{fig:ying rec}, and subsequently to a linear operator, in the same way as in \autoref{sec: lnSem}.

%\section{Conclusion}

% Una programma ricorsivo 

% chiamata ricorsiva quindi corrisponde ad una serie di operatori che agiscono ognuno in un elemento dello spazio di fock, a seconda di quante copie del qubit guardia sono utilizzate.
% Usando questa formulazione Ying definisce un CPO con il quale si può definire il lub della  serie di operatori che rappresentano le chiam

% In particular viene consideratp come spazio del programma un Fock space che è la somma diretta di Hilbert space che corrispondondo al prodotto tensoriale degli spazi delle variabili usate, dove ogni Hilbert space corrisponde 

% In~\cite{MingFoundations, ying2014recursion}, Ying explores this concept of quantum recursion and uses Fock spaces~\cite{attalfock} and second quantization~\cite{bosonfermion} to represent the denotation semantics of a quantum recursive program.

\bibliographystyle{ACM-Reference-Format}
\bibliography{bibl}

\appendix
\section{Omitted Proofs}\label{appx}
\propunind*
\begin{proof}
We proceed by induction.
If $n=0$, $\wt_0 = \Ident$ and $\Wt_0 = \sum_{k=1}^{0}(\dots) + \prod_{i=1}^{0}\prj{1_{t_i}} \cdot \prod_{i=0}^{-1}\unSem{s}\bev(q,n-i) = \Ident$.
Let's consider $\wt_{n+1} = (\prj{0_{t_1}} + \prj{1_{t_1}} \cdot \mathcal{S}(\Wt_{n}) \cdot \unSem{s}) \bev(q,1)$.
By inductive hypothesis $\wt_n = \Wt_n$,
% \begin{equation*}
%      = \sum_{k=1}^{n}(\prod_{i=1}^{k-1}(\prj{1_{t_i}}) \cdot \prj{0_{t_k}} \cdot \prod_{i=n-k}^{n-2}(\bev(q,n-i)\unSem{s}) \cdot \bev(q,1)) + \prod_{i=1}^{n}\prj{1_{t_i}} \cdot \prod_{i=0}^{n-1}\unSem{s}\bev(q,n-i).
% \end{equation*}
thus we need to compute $\mathcal{S}(\Wt_{n})$:
\begin{equation*}
    \mathcal{S}(\Wt_{n}) = \sum_{k=1}^{n}(\prod_{i=1}^{k-1}(\prj{1_{t_{i+1}}}) \cdot \prj{0_{t_{k+1}}} \cdot \prod_{i=n-k}^{n-2}(\bev(q,n-i+1)\unSem{s}) \cdot \bev(q,2)) + \prod_{i=1}^{n}\prj{1_{t_{i+1}}} \cdot \prod_{i=0}^{n-1}\unSem{s}\bev(q,n-i+1).
\end{equation*}
The initial and final values of the first and third products can be updated to remove the $+1$ from the indices, resulting in:
\begin{equation*}
    \mathcal{S}(\Wt_{n}) = \sum_{k=1}^{n}(\prod_{i=2}^{(k+1)-1}(\prj{1_{t_{i}}}) \cdot \prj{0_{t_{k+1}}} \cdot \prod_{i=n-k}^{n-2}(\bev(q,(n+1)-i)\unSem{s}) \cdot \bev(q,2)) + \prod_{i=2}^{n+1}\prj{1_{t_{i}}} \cdot \prod_{i=0}^{n-1}\unSem{s}\bev(q,(n+1)-i).
\end{equation*}
Finally, the $k+1$ can be collected, and the summation indices can be updated, resulting in:
\begin{equation*}
    \mathcal{S}(\Wt_{n}) = \sum_{k=2}^{n+1}(\prod_{i=2}^{k-1}(\prj{1_{t_{i}}}) \cdot \prj{0_{t_{k}}} \cdot \prod_{i=n-(k-1)}^{n-2}(\bev(q,(n+1)-i)\unSem{s}) \cdot \bev(q,2)) + \prod_{i=2}^{n+1}\prj{1_{t_{i}}} \cdot \prod_{i=0}^{n-1}\unSem{s}\bev(q,(n+1)-i).
\end{equation*}
Now we sobstitude this equation in $\wt_{n+1} = \prj{0_{t_1}} \cdot \bev(q,1) + \prj{1_{t_1}} \cdot \mathcal{S}(\Wt_{n}) \cdot \unSem{s} \cdot \bev(q,1)$.
Initially, we compute $\wt' = \prj{1_{t_1}} \cdot \mathcal{S}(\Wt_{n})$, in particular:
\begin{equation*}
    \begin{aligned}
        &\wt' =\! \prj{1_{t_1}} \big( \sum_{k=2}^{n+1}(\prod_{i=2}^{k-1} (\prj{1_{t_{i}}}) \prj{0_{t_{k}}}  \prod_{i=n-(k-1)}^{n-2} (\bev(q,(n+1)-i)\unSem{s}) \bev(q,2)) + \prod_{i=2}^{n+1}\prj{1_{t_{i}}} \prod_{i=0}^{n-1}\unSem{s}\bev(q,(n+1)-i) \big)\\
        &=\! \sum_{k=2}^{n+1}(\prj{1_{t_1}} \prod_{i=2}^{k-1} (\prj{1_{t_{i}}}) \prj{0_{t_{k}}} \prod_{i=n-(k-1)}^{n-2} (\bev(q,(n+1)-i)\unSem{s}) \bev(q,2)) + \prj{1_{t_1}} \prod_{i=2}^{n+1}\prj{1_{t_{i}}} \prod_{i=0}^{n-1}\unSem{s}\bev(q,(n+1)-i)\\
        &=\! \sum_{k=2}^{n+1}(\prod_{i=1}^{k-1} (\prj{1_{t_{i}}}) \prj{0_{t_{k}}} \prod_{i=n-(k-1)}^{n-2} (\bev(q,(n+1)-i)\unSem{s}) \bev(q,2)) + \prod_{i=1}^{n+1}\prj{1_{t_{i}}} \prod_{i=0}^{n-1}\unSem{s}\bev(q,(n+1)-i).\\
    \end{aligned}
\end{equation*}   
Secondly, let's consider $\wt'' = \wt' \cdot \unSem{s} \cdot \bev(q,1)$, it result in:
\begin{equation*}
    \begin{aligned}
        \wt'' &=\! \Big( \sum_{k=2}^{n+1}(\prod_{i=1}^{k-1} (\prj{1_{t_{i}}}) \prj{0_{t_{k}}} \prod_{i=n-(k-1)}^{n-2} (\bev(q,(n+1)-i)\unSem{s}) \bev(q,2))\\
        & + \prod_{i=1}^{n+1}\prj{1_{t_{i}}} \prod_{i=0}^{n-1}\unSem{s}\bev(q,(n+1)-i) \Big) \unSem{s}\bev(q,1)\\
        &=\! \sum_{k=2}^{n+1}(\prod_{i=1}^{k-1} (\prj{1_{t_{i}}}) \prj{0_{t_{k}}} \prod_{i=n-(k-1)}^{n-2} (\bev(q,(n+1)-i)\unSem{s} \big) \bev(q,2) \unSem{s} \bev(q,1))\\
        & + \prod_{i=1}^{n+1}\prj{1_{t_{i}}} \prod_{i=0}^{n-1}\unSem{s}\bev(q,(n+1)-i) \unSem{s} \bev(q,1)
    \end{aligned}
\end{equation*}
If we consider the second and fourth products, we can include $\bev(q,2)\unSem{s}$ by updating the indices.
Specifically, in the first product, when $i=n-2$, we have $n+1-n+2=3$.
Therefore, by setting $n-1$ as the upper limit of the product, we can include $\bev(q,2)\unSem{s}$ in the product.
Similarly, in the fourth product, when $i=n-1$, $\unSem{s} \bev(q,2)$ is included. 
By varying the product from $0$ to $n$, we also include $\unSem{s} \bev(q,1)$.
So finally, we can write $\wt''$ as:
\begin{equation*}
    \wt'' = \sum_{k=2}^{n+1}(\prod_{i=1}^{k-1} (\prj{1_{t_{i}}}) \prj{0_{t_{k}}} \prod_{i=n-(k-1)}^{n-1} (\bev(q,(n+1)-i)\unSem{s}) \bev(q,1)) + \prod_{i=1}^{n+1}\prj{1_{t_{i}}} \prod_{i=0}^{n}\unSem{s}\bev(q,(n+1)-i).
\end{equation*}
Finally, since $\Wt_{n+1} = \prj{0_{t_1}} \bev(q,1) + \Wt''$ and $(\prod_{i=1}^{k-1} (\prj{1_{t_{i}}}) \prj{0_{t_{k}}} \prod_{i=n-(k-1)}^{n-1} (\bev(q,(n+1)-i)\unSem{s}) \bev(q,1))_{k=1} = \prj{0_{t_1}} \cdot \bev(q,1)$ we can write:
\begin{equation*}
    \wt_{n+1} = \sum_{k=1}^{n+1}(\prod_{i=1}^{k-1} (\prj{1_{t_{i}}}) \prj{0_{t_{k}}} \prod_{i=(n+1)-k}^{n-1} (\bev(q,(n+1)-i)\unSem{s}) \bev(q,1)) + \prod_{i=1}^{n+1}\prj{1_{t_{i}}} \prod_{i=0}^{n}\unSem{s}\bev(q,(n+1)-i),
\end{equation*}
i.e., we write $\Wt_{n+1}$ in the form of \autoref{eq: Wnclose}.

\end{proof}

\proplnind*
\begin{proof}
If $n=0$, $\lt_0 = 0$ and $\Lt = \sum_{k=1}^{0}(\dots) = 0$.
By inductive hypothesis,
\begin{equation*}
    \lt_n = \sum_{k=1}^{n}(\prod_{i=1}^{k-1}(\prj{1_{t_i}}) \cdot \prj{0_{t_k}} \cdot \prod_{i=n-k}^{n-2}(\bev(q,n-i)\lnSem{s}) \cdot \bev(q,1)).
\end{equation*}
Let's consider $\lt_{n+1} = (\prj{0_{t_1}} + \prj{1_{t_1}} \cdot \mathcal{S}(\lt_{n}) \cdot \lnSem{s}) \bev(q,1)$.
By inductive hypothesis $\lt_n = \Lt_n$.
First we compute $\mathcal{S}(\Lt_{n})$, in particular,
\begin{equation*}
    \mathcal{S}(\Lt_{n}) = \sum_{k=1}^{n}(\prod_{i=1}^{k-1}(\prj{1_{t_{i+1}}}) \cdot \prj{0_{t_{k+1}}} \cdot \prod_{i=n-k}^{n-2}(\bev(q,n-i+1)\lnSem{s}) \cdot \bev(q,2)).
\end{equation*}
The initial and final values of the first product can be updated to remove the $+1$ from the indices, resulting in:
\begin{equation*}
    \mathcal{S}(\Lt_{n}) = \sum_{k=1}^{n}(\prod_{i=2}^{(k+1)-1}(\prj{1_{t_{i}}}) \cdot \prj{0_{t_{k+1}}} \cdot \prod_{i=n-k}^{n-2}(\bev(q,(n+1)-i)\lnSem{s}) \cdot \bev(q,2)).
\end{equation*}
Then, the $k+1$ can be collected, and the summation indices can be updated, resulting in the following:
\begin{equation*}
    \mathcal{S}(\Lt_{n}) = \sum_{k=2}^{n+1}(\prod_{i=2}^{k-1}(\prj{1_{t_{i}}}) \cdot \prj{0_{t_{k}}} \cdot \prod_{i=n-(k-1)}^{n-2}(\bev(q,(n+1)-i)\lnSem{s}) \cdot \bev(q,2)).
\end{equation*}
Now we sobstitude this equation in $\lt_{n+1} = \prj{0_{t_1}} \cdot \bev(q,1) + \prj{1_{t_1}} \cdot \mathcal{S}(\Lt_{n}) \cdot \lnSem{s} \cdot \bev(q,1)$.
Initially, we compute $\lt' = \prj{1_{t_1}} \cdot \mathcal{S}(\Lt_{n})$, in particular:
\begin{equation*}
    \begin{aligned}
        &\lt' =\! \prj{1_{t_1}} \big( \sum_{k=2}^{n+1}\prod_{i=2}^{k-1} (\prj{1_{t_{i}}}) \prj{0_{t_{k}}}  \prod_{i=n-(k-1)}^{n-2} (\bev(q,(n+1)-i)\lnSem{s}) \bev(q,2) \big)\\
        &=\! \sum_{k=2}^{n+1}(\prj{1_{t_1}} \prod_{i=2}^{k-1} (\prj{1_{t_{i}}}) \prj{0_{t_{k}}} \prod_{i=n-(k-1)}^{n-2} (\bev(q,(n+1)-i)\lnSem{s}) \bev(q,2))\\
        &=\! \sum_{k=2}^{n+1}(\prod_{i=1}^{k-1} (\prj{1_{t_{i}}}) \prj{0_{t_{k}}} \prod_{i=n-(k-1)}^{n-2} (\bev(q,(n+1)-i)\lnSem{s}) \bev(q,2)).\\
    \end{aligned}
\end{equation*}   
Secondly, let's consider $\lt'' = \lt' \cdot \lnSem{s} \cdot \bev(q,1)$, it result in:
\begin{equation*}
    \begin{aligned}
        \lt'' &=\! \Big( \sum_{k=2}^{n+1}\prod_{i=1}^{k-1} (\prj{1_{t_{i}}}) \prj{0_{t_{k}}} \prod_{i=n-(k-1)}^{n-2} (\bev(q,(n+1)-i)\lnSem{s}) \bev(q,2) \Big) \lnSem{s}\bev(q,1) =\\
        &=\! \sum_{k=2}^{n+1}\prod_{i=1}^{k-1} (\prj{1_{t_{i}}}) \prj{0_{t_{k}}} \prod_{i=n-(k-1)}^{n-2} (\bev(q,(n+1)-i)\lnSem{s} \big) \bev(q,2) \lnSem{s} \bev(q,1).\\
    \end{aligned}
\end{equation*}

If we consider the second product, we can include $\bev(q,2)\lnSem{s}$ by updating the indices.
Specifically, when $i=n-2$, $n+1-n+2=3$, therefore, by setting $n-1$ as the upper limit of the product, we can include $\bev(q,2)\lnSem{s}$ in the product.
So finally, we can write $\lt''$ as:
\begin{equation*}
    \lt'' = \sum_{k=2}^{n+1}\prod_{i=1}^{k-1} (\prj{1_{t_{i}}}) \prj{0_{t_{k}}} \prod_{i=n-(k-1)}^{n-1} (\bev(q,(n+1)-i)\lnSem{s}) \bev(q,1).
\end{equation*}
Finally, since $\lt_{n+1} = \prj{0_{t_1}} \bev(q,1) + \lt''$ and $(\prod_{i=1}^{k-1} (\prj{1_{t_{i}}}) \prj{0_{t_{k}}} \prod_{i=n-(k-1)}^{n-1} (\bev(q,(n+1)-i)\lnSem{s}) \bev(q,1))_{k=1} = \prj{0_{t_1}} \cdot \bev(q,1)$ we can write:
\begin{equation*}
    \lt_{n+1} = \sum_{k=1}^{n+1}\prod_{i=1}^{k-1} (\prj{1_{t_{i}}}) \prj{0_{t_{k}}} \prod_{i=(n+1)-k}^{n-1} (\bev(q,(n+1)-i)\lnSem{s}) \bev(q,1),
\end{equation*}
i.e., we have writen $\lt_{n+1}$ in the form of \autoref{eq: Lnclose}.

\end{proof}

\end{document}